\documentclass[11pt,letterpaper]{article}

\usepackage[top=1in, bottom=1in, left=1in, right=1in]{geometry}

\usepackage{microtype}
\usepackage{graphicx}
\usepackage{subfigure}
\usepackage{booktabs} 

\usepackage[draft]{hyperref}

\usepackage{amsmath}
\usepackage{amsthm}
\usepackage{amsfonts}
\usepackage{amssymb}
\usepackage{bbm}
\usepackage{mathabx} 
\usepackage{mathrsfs}
\usepackage{graphicx}
\usepackage{color}
\usepackage[inline]{enumitem}
\usepackage{bm}

\newcommand{\E}[1]{\textup{E}\big( #1 \big)}

\newcommand{\bH}{\mb{H}}
\newcommand{\bHhat}{\widehat{\mb{H}}}
\newcommand{\res}[2]{\big[ #1 \big]_{#2}}

\newcommand{\T}{\intercal}
\newcommand{\supp}{\textup{S}_i}

\newcommand{\budget}{\textup{b}}
\newcommand{\prob}{\textup{Pr}}

\newcommand{\real}{\mathbb{R}}
\newcommand{\graph}{\mathcal{G}}
\newcommand{\nodes}{\mathcal{V}}
\newcommand{\edges}{\mathcal{E}}
\newcommand{\actions}{\mathcal{A}}
\newcommand{\game}{\mathscr{G}}
\newcommand{\nash}{\textup{NE}}
\newcommand{\sign}{\textup{sign}}

\newcommand{\mb}[1]{\mathbf{#1}}
\newcommand{\minus}{\text{\normalfont -}}
\newcommand{\seq}[1]{\{ 1,\dots, #1 \}}
\newcommand{\inner}[2]{\langle #1 \,, #2 \rangle}

\newcommand{\bbeta}{\bm{\beta}}
\newcommand{\bpsi}{\bm{\psi}}
\newcommand{\bgamma}{\bm{\gamma}}
\newcommand{\bz}{\mb{z}}
\newtheorem{theorem}{Theorem} 
 
\newtheorem{lemma}{Lemma}
\newtheorem{assumption}{Assumption}

\title{Provable Sample Complexity Guarantees for Learning of Continuous-Action Graphical Games with Nonparametric Utilities}
\date{}

\author{%
	Adarsh Barik \\
	Department of Computer Science\\
	Purdue University\\
	West Lafayette, Indiana, USA\\
	\texttt{abarik@purdue.edu} \\
	\and
	Jean Honorio \\
	Department of Computer Science \\
	Purdue University \\
	West Lafayette, Indiana, USA\\
	\texttt{jhonorio@purdue.edu} \\
}
\begin{document}
\maketitle

\begin{abstract}
In this paper, we study the problem of learning the exact structure of continuous-action games with non-parametric utility functions. We propose an $\ell_1$ regularized method which encourages sparsity of the coefficients of the Fourier transform of the recovered utilities. Our method works by accessing very few Nash equilibria and their noisy utilities. Under certain technical conditions, our method also recovers the exact structure of these utility functions, and thus, the exact structure of the game. Furthermore, our method only needs a logarithmic number of samples in terms of the number of players and runs in polynomial time. We follow the primal-dual witness framework to provide provable theoretical guarantees.
\end{abstract}

\section{Introduction}
\label{sec:introduction}
Game theory has been extensively used as a framework to model and study the strategic interactions amongst rational but selfish individual players who are trying to maximize their payoffs. Game theory has been applied in many fields including but not limited to social and political science, economics, communication, system design and computer science. In non-cooperative games each player decides its action based on the actions of others players. These games are characterized by the equilibrium solution concept such as \emph{Nash equilibrium (NE)}~\cite{Nash51} which serves a descriptive role of the stable outcome of the overall behavior of self-interested players (e.g., people, companies, governments, groups or autonomous systems) interacting strategically with each other in distributed settings. 

Graphical games, introduced within the AI community about two decades ago, \emph{graphical games}~\cite{Kearns01}, are a representation of multiplayer games which capture and
exploit locality or sparsity of direct influences. They are most appropriate for large-scale
population games in which the payoffs of each player are determined by the actions
of only a small number of other players. Indeed, graphical games played a prominent role in establishing the computational complexity of computing NE in normal-form games as well as in succinctly representable multiplayer games (see, e.g.,~\cite{Daskalakis06,Daskalakis09,Daskalakis09b} and the references therein). Graphical games have been studied for both discrete and continuous actions.

\paragraph{Inference in Graphical Games.}
There has been a large body of work on computing classical equilibrium solution concepts. 
The Nash equilibria and \emph{correlated equilibria}~\cite{Aumann74} in graphical games have been studied by ~\cite{Blum06,Kearns01,Kakade03,Ortiz02,Papadimitriou08,Vickrey02,facchinei2007generalized,scutari2010convex,perkins2013asynchronous,perkins2015mixed,mertikopoulos2019learning}. The computation of the \emph{price of anarchy} was studied in \cite{Benzwi11}. In addition, \cite{Irfan14} identified the most influential players, i.e., a small set of players whose collective behavior forces every other player to a unique choice of action. All the works above focus on inference problems for graphical games, and fall in the field of algorithmic game theory. All of these methods assume access to graphical game network and payoffs of the games under consideration.

\paragraph{Learning Graphical Games.}
Consider an example of the following decision problem in a marketplace where there are $n$ producers selling $k$ products. Each producer decides the production quantity for a product based on its supply and demand in the market which could be treated as the action vector for the producer. Each producer gauges the supply and demand by looking at the production quantity of few other players which it considers its \emph{close competitors}. The payoff for a producer can be measured by its revenue. In order to answer the \emph{inference} problems discussed above, we would need to recover the structure of this graphical game. In particular we ask: Given that we have access to a  noisy observation of the player's revenue and few joint action vectors at equilibria, can we learn the neighbors of a player? Of course, once we learn the local structure for each individual player then we can combine them in order to obtain the entire structure of the graphical game. Learning the structure of a game is essential to the development, potential use and success of game-theoretic models in large-scale practical applications. In this paper, we study the problem of learning the graph structure in a continuous-action graphical game with non-parametric utility functions.

\paragraph{Related Work.}
None of the prior literature has dealt with either continuous actions or non-parametric utilities. In discrete-action games, \cite{Honorio15} proposed a maximum-likelihood approach to learn “linear influence games” - a class of parametric graphical games with binary actions and linear payoffs.  However, their method runs in exponential time and the authors assumed a specific observation model for the strategy profiles. For the same specific observation model, \cite{Ghoshal16} proposed a polynomial time algorithm, based on $\ell_1$-regularized logistic regression, for learning linear influence games. Their strategy profiles (or joint actions) were drawn from a mixture of uniform distributions: one over the pure-strategy Nash equilibria (PSNE) set, and the other over its complement. \cite{Ghoshal17} obtained necessary and sufficient conditions for learning linear influence games under arbitrary observation model. \cite{Garg16} use a discriminative, max-margin based approach, to learn tree structured polymatrix games\footnote{Polymatrix games are graphical games where each player's utility is a sum of unary (single player) and pairwise (two players) potential functions.}. Their method runs in exponential time and the authors show that learning polymatrix games is NP-hard under this max-margin setting, even when the class of graphs is restricted to trees. Finally, \cite{Ghoshal18} proposed a polynomial time algorithm for learning sparse polymatrix games in the discrete-action setting.

\paragraph{Contribution.}
We propose a novel method to learn graphical games with non-parametric payoffs. To this end, we propose an $\ell_1$-norm regularized method to learn graphical games, which encourages sparsity of the coefficients of the Fourier transform of the recovered utilities. For $n$ players, at most $d$ in-neighbors per player, we show that $N = \mathcal{O}(d^3 \log (n))$  \emph{samples} are sufficient to recover the exact structure of the true game. We also fully characterize a utility function which is at most $\epsilon$ away from the true utility function. We achieve the following goals in this paper: 

\begin{enumerate}
	\item {\bf Correctness -} Our method correctly recovers the true structure of the graphical games.
	\item {\bf Computational efficiency -} Our method has a polynomial time complexity with respect to the number of players and thus, can handle high dimensional cases.
	\item {\bf Sample complexity -} Our method achieves only a logarithmic sample complexity with respect to the number of players.
\end{enumerate}

\section{Preliminaries}
\label{sec:preliminaries}
In this section, we introduce our notations and formally setup our problem.

\subsection{Notation}
\label{subsec:notation}
Consider a directed graph $\graph(\nodes,\edges)$, where $\nodes$ and $\edges$ are set of vertices and edges respectively. We define $\nodes \triangleq \seq{n}$, where each vertex corresponds to one player. We denote the in-neighbors of a player $i$ by $\supp$, i.e., $\supp = \{ j \mid (j, i) \in \edges \}$ (i.e., the set of nodes that point to node $i$ in the graph). All the other players are denoted by $\supp^c$, i.e., $\supp^c = \seq{n} \backslash (\supp \cup \{i\}) $. Let $|\supp| \leq d$ and $|\supp^c| \geq n-d$. 

For each player $i \in \nodes$, there is a set of actions or \emph{pure-strategies}  $\actions_i$. That is, player $i$ can take action $x_i \in \actions_i$. Each action $x_i$ consists of making $k$ decisions on a limited budget $\budget \in \real$. We consider games with continuous actions. Mathematically, $x_i \in \actions_i = \real^k$ and $\| x_i \|_2 \leq \budget$. We use $x_{\minus i}$ to denote the collection of actions of all players but $i$. For each player $i$, there is also a local payoff function $u_i : \actions_i \times ( \bigtimes_{j \in \supp} \actions_j ) \to \real$ mapping the joint action of player $i$ and its in-neighbors $\supp$, to a real number. A joint action $\mb{x}^* \in \bigtimes_{i \in \nodes} \actions_i$ is a \emph{pure-strategy Nash equilibrium (PSNE)} of a graphical game iff, no player $i$ has any incentive to unilaterally deviate from the prescribed action $x_i^* \in \actions_i$, given the joint action of its in-neighbors $x_{\supp}^* \in \bigtimes_{j \in \supp} \actions_j$ in the equilibrium. We denote a game by $\game$, and the set of all PSNE is denoted by  $\nash(\game)$.

For a matrix $\mb{A} \in \real^{p \times q}$ and two sets $S \subseteq \seq{p}$ and $T \subseteq \seq{q}$,  $\mb{A}_{S T}$ denotes $\mb{A}$ restricted to rows in $S$ and columns in $T$. Similarly, $\mb{A}_{S.}$ and $\mb{A}_{.T}$ are row and column restricted matrices respectively. For a vector $\textbf{m} \in \real^q$, the $\ell_{\infty}$-norm is defined as $ \| \textbf{m} \|_{\infty} = \max_{i \in \seq{p}} | \textbf{m}_i | $. 
The $\ell_{\infty}$-operator norm for $\mb{A}$ is defined as $\| \mb{A} \|_{\infty, \infty} = \max_{i \in \seq{p}}\sum_{j=1}^q | \mb{A}_{ij} |$.
The spectral norm of $\mb{A}$ is defined as $\| \mb{A} \|_{2} = \sup_{ \|\mb{x}\|_2 = 1} \|\mb{A} \mb{x}\|_2$.

\subsection{Model}
\label{subsec:model}
Now we describe the basic setup of our problem. We consider a graphical game with the property that utility function of each player $i$ is decomposable into sum of pairwise functions which only depend on the in-neighbors of $i$. In particular,
\begin{align}
\label{eq:pairwise utility}
u_i(x) = \sum_{j \in \supp} u_{ij}(x_i, x_j)
\end{align}
where $\supp$ is the set of neighbors of player $i$. 

For two real valued functions $f(x)$ and $g(x)$ on the domain $x \in \mathcal{X}$, we define their inner product as
\begin{align}
\begin{split}
	\inner{f}{g} = \int_{x \in \mathcal{X}} f(x) g(x) dx \ .
\end{split}
\end{align}  
We denote $L_2$ norm of $f(x)$ as
\begin{align}
\begin{split}
	\| f \|_2 = \sqrt{\int_{x \in \mathcal{X}} f^2(x) dx} \ .
\end{split}
\end{align} 
Using the above definitions, it is useful to express our model in the basis functions. Let $\psi_k(.,.), \forall k=\{1,\dots,\infty\}$ be a set of uniformly bounded, orthonormal basis functions such that 
\begin{align}
u_{ij}(x_i, x_j) = \sum_{k=0}^{\infty} \beta_{ijk}^* \psi_{k}(x_i, x_j)
\end{align}
where $\beta_{ijk}^* = \inner{u_{ij}}{\psi_k}$. We assume that for all $i, j \in \seq{n}, i \ne j$, the weight magnitudes $|\beta_{ijk}^*|, \forall k \in \{r+1,\dots,\infty\}$ form a convergent series for a sufficiently large $r$. For example, for large enough $r$, Fourier coefficients of a periodic and twice continuously differentiable function form a convergent series \cite{rust2013convergence}.  For all $i,j \in \{1,\dots,n\}$ we assume that $\sup_{x_i \in \actions_i, x_j \in \actions_j} |\psi_k(x_i, x_j)| \leq \overline{\psi}$. We assume that we have access to a set of joint actions  $\mb{x}$ and the corresponding noisy payoff for each player. 

We define a local setup where we learn the in-neighbors of a player by having access to its perturbed payoff and perturbed basis functions. This allows us to apply our method in a distributed setting where in-neighbors for each node are recovered independently and combined later to get the complete structure of the graphical game.  

\subsection{Sampling Mechanism}
Treating the outcomes of the game as ``samples'' observed across multiple ``plays'' of the same game is a recurring theme in the literature for learning games \cite{Honorio15,Ghoshal17,Ghoshal18}. All of these works assume access to Nash equilibria and a noise mechanism. Noise could be added to each player's strategy at a local level or by mixing the sample with a non-Nash equilibria set at a global level. Since none of the prior literature has dealt with either continuous actions or non-parametric utilities, we propose a novel sampling mechanism. We assume that we have noisy black-box access to payoffs. For a joint action $x$ the blackbox outputs a noisy payoff $\widetilde{u}_i(x)$. The blackbox computes this noisy payoff by first computing noisy basis function values $\widetilde{\psi}_k(x_i, x_j)$ and then finally taking their weighted sum to get the final noisy payoff.

Mathematically,   
\begin{align}
\widetilde{\psi}_k(x_i, x_j) = \psi_k(x_i, x_j) + \gamma_k(x_i, x_j)
\end{align}  
where $\gamma_k(x_i, x_j)$ are independent zero mean sub-Gaussian noise with variance proxy $\sigma^2$. The class of sub-Gaussian random variables includes for instance Gaussian random variables, any bounded random variable (e.g. Bernoulli, multinomial, uniform), any random variable with strictly log-concave density, and any finite mixture of sub-Gaussian variables. Thus, using sub-Gaussian noise makes our setup quite general in nature. Correspondingly, the resulting noisy payoff function that we observe can be written as follows:
\begin{align}
\widetilde{u}_i(x) = \sum_{j \in \supp} \widetilde{u}_{ij}(x_i, x_j) =  \sum_{j \in \supp} \sum_{k=0}^\infty \beta_{ijk}^* \widetilde{\psi}_k(x_i, x_j)
\end{align}
where 
\begin{align}
\widetilde{u}_{ij}(x_i, x_j) = \sum_{k=0}^\infty \beta_{ijk}^* \widetilde{\psi}_k(x_i, x_j) \ .
\end{align}

\paragraph{Estimation.}
Let $\overline{u}_{ij}(x_i, x_j)$ be our estimation of $u_{ij}(x_i, x_j)$. Ideally, we would like to estimate an infinitely long vector with entries indexed by $(j,k)$ for $j \in \seq{n}, j \ne i$ and $k \in \seq{\infty}$ for every player $i$. However, this is impractical with finite computing resources. Rather, we estimate an $(n -1) \times r$ dimensional vector $\bbeta$ with entries indexed by $(j,k)$ for $j \in \seq{n}, j \ne i$ and $k \in \seq{r}$ for every player $i$. Using these finite number of coefficients, we estimate pairwise utility in the following manner,
\begin{align}
\overline{u}_{ij}(x_i, x_j) = \sum_{k=0}^r \beta_{ijk} \psi_k(x_i, x_j)
\end{align} 
and thus, the resulting estimation of the utility function is:
\begin{align}
\overline{u}_i(x) = \sum_{j \in \supp} \sum_{k=0}^r \beta_{ijk} \psi(x_i, x_j)
\end{align}

\section{Theoretical Results}
\label{sec:theoretical results}
In this section, we setup our estimation problem and prove our theoretical results. We also mention some technical assumptions in this section which we will use in our proofs. Let $D$ be a collection of $N$ independent samples. We estimate $\bbeta^*$ by solving the following optimization problem in sample setting for some $\lambda > 0$:
\begin{align}
\label{eq:estimateprob}
\begin{split}
&\min_{\beta} \frac{1}{N} \sum_{x \in D} (\widetilde{u}_i(x) - \overline{u}_i(x))^2 + \lambda \sum_{j \ne i} \sum_{k=0}^r |\beta_{ijk}| \\
&\text{such that } \sum_{j \ne i} \sum_{k=0}^r |\beta_{ijk}|  \leq C
\end{split}
\end{align}

where $C > 0$ is a constant that acts as a budget on the coefficients $\beta_{ijk}$ to ensure that they are not unbounded. We would like to prove that $\beta_{ijk} = 0, \forall j \notin \supp$ and $\beta_{ijk} \ne 0, \forall j \in \supp$. This gives us a straight forward way of picking in-neighbors of player $i$ by solving optimization problem \eqref{eq:estimateprob}. We use an auxiliary variable $w = \sum_{j \ne i}  \sum_{k=0}^r |\beta_{ijk}|$ and prove the following lemma to get Karush-Kuhn-Tucker (KKT) conditions at the optimum.

\begin{lemma}[KKT conditions]
	\label{lem:kkt}
	The following Karush-Kuhn-Tucker (KKT) conditions hold at the optimal solution of the optimization problem \eqref{eq:estimateprob}:
	\begin{align}
	\begin{split}
	\text{Stationarity: } &\frac{\partial}{\partial \bbeta} \big[ \frac{1}{N} \sum_{x \in D} (\widetilde{u}(x) - \overline{u}(x))^2 \big] + \lambda \bz = 0 \\
	\text{Primal Feasibility: } & w \leq C \\
	& w =\sum_{j \ne i}  \sum_{k=0}^r |\beta_{ijk}|  \\
	\text{Dual Feasibility: } & \lambda \geq 0
	\end{split}
	\end{align} 
	where $\bz \in \real^{r(n-1) \times 1} $indexed by $(j,k), j \in \seq{n}, j \ne i, k \in \seq{r}$. It is defined as follows:
	\begin{align}
	z_{ijk} = \begin{cases}
	\sign(\beta_{ijk}), \text{ if } \beta_{ijk} \ne 0 \\
	[-1, 1],  \text{ otherwise }
	\end{cases}
	\end{align}	
\end{lemma}
(Check Appendix \ref{proof:kkt} for details.)

\subsection{Assumptions} 
\label{subsec:assumptions}
In this subsection, we will discuss the key technical assumptions which are required for our theoretical results. For notational clarity, we define the following two quantities:
\begin{align}
\begin{split}
\bH &\triangleq \frac{1}{N} \sum_{x \in D} \bpsi(x) \bpsi(x)^\T \\
\bHhat &\triangleq \frac{1}{N} \sum_{x \in D} \big( \bpsi(x) \bpsi(x)^\T + \bpsi(x) \bgamma(x)^\T \big)
\end{split}
\end{align}
It should be noted that $\bHhat = \E{\bH}$. For our first assumption, we want $\bH$ restricted to in-neighbors to be an invertible matrix. Formally, 
\begin{assumption}
	\label{assum:positive definiteness}
	We assume that $\bH$ is positive definite matrix, i.e., $\Lambda_{\min} (\res{\bH}{\supp\supp}) = C_{\min} > 0$ where $\Lambda_{\min}$ denotes the minimum eigenvalue.
\end{assumption} 
This assumption readily implies that $\bHhat$ is invertible with sufficient number of samples $N$. More technically, 
\begin{lemma}
	\label{lem:positive definiteness}
	If   $\Lambda_{\min} (\res{\bH}{\supp\supp}) = C_{\min} > 0$ then  $\Lambda_{\min} (\res{\bHhat}{\supp\supp}) = C_{\min} - \epsilon$ for some $\epsilon > 0$ with probability at least $\exp(\frac{-N\epsilon^2}{2 \budget^2 \sigma^2} + \mathcal{O}(rd) )$.
\end{lemma}  
(Check Appendix \ref{proof:positive definiteness} for details.)

The above lemma ensures that our optimization problem \eqref{eq:estimateprob} has a unique solution and all our theoretical guarantees hold for this particular solution. As for our second assumption, we require that non-in-neighbors of a player do not affect the player's action too much. We achieve this by proposing a mutual incoherence assumption.
\begin{assumption}
	\label{assum:mutual incoherence}
	We assume that $\|| \res{\bH}{\supp^c\supp} \res{\bH}{\supp\supp}^{-1}\||_{\infty} \leq 1 - \alpha$ for some $0 < \alpha < 1$.
\end{assumption}
As with Assumption \ref{assum:positive definiteness}, it can be shown easily that with enough number of samples Assumption \ref{assum:mutual incoherence} is satisfied in sample setting with sufficient number of samples $N$.
\begin{lemma}
	\label{lem:mutual incoherence}
	If $\|| \res{\bH}{\supp^c\supp} \res{\bH}{\supp\supp}^{-1}\||_{\infty} \leq 1 - \alpha$ for some $0 < \alpha < 1$ then $\|| \res{\bHhat}{\supp^c\supp} \res{\bHhat}{\supp\supp}^{-1}\||_{\infty} \leq 1 - \frac{\alpha}{2}$ with probability at least $1 - \exp(\frac{-K NC_{\min}^2 \alpha^2}{(1-\alpha)^2 r^3 d^3 \widebar\psi^2\sigma^2} + \log(2r^2(n-d)d)) - \exp(\frac{-N \alpha}{6 r^2 d^2 \widebar\psi^2\sigma^2} + \log(2r^2(n-d)d)) - \exp(\frac{-2N}{\budget^2\sigma^2C_{\min}^2} + \mathcal{O}(rd)) - 2 \exp(\frac{-N\alpha^2}{24r^3d^3\widebar\psi^2\sigma^2} + \log(r^2d^2)) $ for a constant $K>0$.
\end{lemma} 
(Check Appendix \ref{proof:mutual incoherence} for details.)

While mutual incoherence is new to graphical games with continuous actions and non-parametric utilities, it has been a standard assumption in various estimation problems such as compressed sensing~\cite{wainwright2009sharp}, Markov random fields~\cite{ravikumar2010high}, non-parametric regression~\cite{ravikumar2007spam}, diffusion networks~\cite{daneshmand2014estimating}, among others.

\begin{assumption}
	\label{assum:min wieght}
	We assume that $\min_{j \in \supp} |\beta_{ijk}^*| > \Delta(\lambda, C_{\min}, d)$ where 
	\begin{align*}
	&\Delta(\lambda, C_{\min}, d) = \frac{2 \lambda}{C_{\min}} ( 2\sqrt{d} +\frac{1} {\sqrt{|\supp |}} )
	\end{align*}
\end{assumption}
The minimum weight assumption is a standard practice in the literature employing the primal-dual witness technique \cite{wainwright2009sharp,ravikumar2010high}. This assumption ensures that the coefficients $\beta_{ijk}^*$ are not arbitrarily close to zero which can make inference very difficult for any method. 

\subsection{Primal-Dual Witness Method}
\label{subsec:primal-dual}
We use primal-dual witness framework \cite{wainwright2009sharp} to prove our theoretical results. We note that $\beta_{ijk}^* = 0, \forall j \ne \supp$. We also assume that $\beta_{ijk} = 0, \forall j \ne \supp$. As first step of our proof, we will justify this choice. Then we bound the Euclidean norm distance between the estimated $\bbeta$ and the true $\bbeta^*$. Subsequently, we show that if the non-zero entries in $\bbeta^*$ satisfy a minimum weight criteria then the recovered $\bbeta$ matches $\bbeta^*$ up to its sign. This allows us to identify the in-neighbors for each player. 
\begin{theorem}
	\label{thm:main theorem}
	Consider a continuous-action graphical game $\game$ such that Assumptions \ref{assum:positive definiteness}, \ref{assum:mutual incoherence} and \ref{assum:min wieght} are satisfied for each player and payoff function for each player is decomposable according to Equation \eqref{eq:pairwise utility}.  Let $\lambda > g(\overline{\psi}, n, \delta, N, d, C, \alpha)$ and $N = \mathcal{O}(\frac{\widebar{\psi}^2 \sigma^2 C^2 \delta r^3d^3}{\epsilon^2}\log (rn))$, then we prove the following claims by solving the optimization problem \eqref{eq:estimateprob}.  
	\begin{enumerate}
			\item We recover the correct set of non-neighbors for each player $i$.
			\item We recover the exact structure of the graphical game $\game$. 
			\item Furthermore, we estimate a payoff function which is $\epsilon$ close to the true payoff by by estimating $\bbeta$ for each player $i$ by solving the optimization problem \eqref{eq:estimateprob}.
	\end{enumerate} 	
	where $C_{\min}$  is the minimum eigenvalue of $\res{\bH}{\supp\supp}$ and $\delta$ is an arbitrary parameter which depends on $r$. 
		
	The functions $g(\cdot)$ is defined as follows:
	\begin{align*}
		&g(\overline{\psi}, n, \delta, N, d, C, \alpha) = \max\big( 20 \sqrt{2} \overline\psi \sqrt{\frac{d\delta\log d}{N}}, 10 \sqrt{2} \sigma \overline\psi C\sqrt{\frac{\log d}{N}},\frac{20 \sqrt{2}\overline\psi }{1 - \frac{\alpha}{2}}\sqrt{\frac{d\delta\log (n-d)}{N}},  \\
		&\frac{10 \sqrt{2} \sigma \overline\psi C }{1 - \frac{\alpha}{2}} \sqrt{\frac{\log (n-d)}{N}}  \big)
	\end{align*}
\end{theorem}

By doing some simple algebraic manipulation it is easy to see that, 
\begin{align}
\begin{split}
&\frac{1}{N} \sum_{x \in D} (\widetilde{u}_i(x) - \overline{u}_i(x))^2=\frac{1}{N} \sum_{x \in D} ( (\bbeta^* - \bbeta )^\T \widetilde{\bpsi} +  \sum_{j \ne i} \sum_{k=r+1}^\infty \beta_{ijk}^* \widetilde{\psi}_k(x_i, x_j) + \bbeta^\T \bgamma(x))^2 
\end{split}
\end{align}
Here, $\bbeta^* \in \real^{r(n-1) \times 1}, \widetilde{\bpsi} \in \real^{r(n-1) \times 1} $ and $\bgamma \in \real^{r(n-1) \times 1}$ are indexed by $(j,k), j \in \seq{n}, j \ne i, k \in \seq{r}$. 
Thus Stationarity KKT condition of Equation \eqref{eq:kkt} becomes:
\begin{align*}
\begin{split}
&\frac{2}{N} \sum_{x \in D} (\bgamma - \widetilde{\bpsi}(x))(  \widetilde{\bpsi}(x)^\T (\bbeta^* - \bbeta ) + \sum_{j \ne i} \sum_{k=r+1}^\infty \beta_{ijk}^* \widetilde{\psi}_k(x_i, x_j) + \bgamma^\T(x)  \bbeta)  + \lambda \bz = 0
\end{split}
\end{align*}
Let $\bbeta = \begin{bmatrix} \bbeta_{\supp} \\ \bbeta_{\supp^c} \end{bmatrix}$, where $\bbeta_{\supp}$ contains the entries $\beta_{ijk} \ne 0$ and $\bbeta_{\supp^c}$ contains the entries $\beta_{ijk} = 0$. From the above, we have
\begin{align*}
\begin{split}
&\frac{2}{N} \sum_{x \in D} (\bgamma - \widetilde{\bpsi}(x))(  \widetilde{\bpsi}_{\supp}(x)^\T (\bbeta_{\supp}^* - \bbeta_{\supp} ) + \sum_{j \ne i} \sum_{k=r+1}^\infty \beta_{ijk}^* \widetilde{\psi}_k(x_i, x_j) + \bgamma_{\supp}^\T(x) \bbeta_{\supp})  + \lambda \bz = 0
\end{split}
\end{align*}
Separating the indices for the support (in-neighbors) and the non-support, we can write the above equation in two parts. The first for the indices in the support,
\begin{align}
\label{eq:support}
\begin{split}
&\frac{2}{N} \sum_{x \in D} - \bpsi_{\supp}(x)(  \widetilde{\bpsi}_{\supp}(x)^\T (\bbeta_{\supp}^* - \bbeta_{\supp} ) + \sum_{j \ne i} \sum_{k=r+1}^\infty \beta_{ijk}^* \widetilde{\psi}_k(x_i, x_j) + \bgamma_{\supp}^\T(x) \bbeta_{\supp})  + \lambda \bz_{\supp} = 0
\end{split}
\end{align}
and the second for the indices in the non-support,
\begin{align}
\label{eq:nonsupport}
\begin{split}
&\frac{2}{N} \sum_{x \in D} - \bpsi_{\supp^c}(x)(  \widetilde{\bpsi}_{\supp}(x)^\T (\bbeta_{\supp}^* - \bbeta_{\supp} ) +  \sum_{j \ne i} \sum_{k=r+1}^\infty \beta_{ijk}^* \widetilde{\psi}_k(x_i, x_j) + \bgamma_{\supp}^\T(x) \bbeta_{\supp})  + \lambda \bz_{\supp^c} = 0 \ .
\end{split}
\end{align} 
We make use of Assumption \ref{assum:positive definiteness} and rearrange Equation \eqref{eq:support} to get,
\begin{align}
\label{eq:delbeta}
\begin{split}
&\bbeta_{\supp}^* - \bbeta_{\supp} =  - (\sum_{x \in D}  \bpsi_{\supp}(x) \widetilde{\bpsi}_{\supp}(x)^\T)^{\minus 1}  (\sum_{x \in D} \bpsi_{\supp}(x) \sum_{j \ne i} \sum_{k=r+1}^\infty \beta_{ijk}^* \widetilde{\psi}_k(x_i, x_j)) -( \sum_{x \in D}  \bpsi_{\supp}(x)  \widetilde{\bpsi}_{\supp}(x)^\T)^{\minus 1} \\
&(\sum_{x \in D} \bpsi_{\supp}(x) \bgamma_{\supp}^\T(x) \bbeta_{\supp}) + \lambda ( \sum_{x \in D} \bpsi_{\supp}(x)  \widetilde{\bpsi}_{\supp}(x)^\T)^{\minus 1} \bz_{\supp}
\end{split}
\end{align}
Substituting Equation \eqref{eq:delbeta} in Equation \eqref{eq:nonsupport} and rearranging the terms we get,
\begin{align*}
\begin{split}
&\lambda \bz_{\supp^c}  = - \frac{2}{N} \res{\bHhat}{\supp^c\supp}  \res{\bHhat}{\supp\supp}^{\minus 1} (\sum_{y \in D} \bpsi_{\supp}(y) \sum_{j \ne i} \sum_{k=r+1}^\infty  \beta_{ijk}^* \widetilde{\psi}(y_i, y_j) ) - \frac{2}{N} \res{\bHhat}{\supp^c\supp}  \res{\bHhat}{\supp\supp}^{\minus 1} (\sum_{y \in D} \bpsi_{\supp}(y) \bgamma_{\supp}^\T(y)  \bbeta_{\supp}) +\\
& \lambda \frac{2}{N} \res{\bHhat}{\supp^c\supp}  \res{\bHhat}{\supp\supp}^{\minus 1} \bz_{\supp}  + \frac{2}{N} \sum_{x \in D} \bpsi_{\supp^c}(x) \sum_{j \ne i} \sum_{k=r+1}^\infty \beta_{ijk}^* \widetilde{\psi}_k(x_i, x_j) + \frac{2}{N} \sum_{x \in D} \bpsi_{\supp^c}(x) \bgamma_{\supp}^\T(x) \bbeta_{\supp}  
\end{split}
\end{align*} 
Using the triangle norm inequality and noting that $\|Ab\|_{\infty} \leq |||A|||_{\infty} \|b\|_{\infty}$, we obtain
\begin{align}
\begin{split}
&  \lambda \| \bz_{\supp^c} \|_{\infty}  \leq  \frac{2}{N} ||| \res{\bHhat}{\supp^c\supp}  \res{\bHhat}{\supp\supp}^{\minus 1}|||_{\infty} \| (\sum_{y \in D} \bpsi_{\supp}(y) \sum_{j \ne i} \sum_{k=r+1}^\infty \beta_{ijk}^* \widetilde{\psi}(y_i, y_j)) \|_{\infty} + \frac{2}{N} ||| \res{\bHhat}{\supp^c\supp}  \res{\bHhat}{\supp\supp}^{\minus 1} |||_{\infty} \\
& \| (\sum_{y \in D} \bpsi_{\supp}(y) \bgamma_{\supp}^\T(y)  \bbeta_{\supp}) \|_{\infty} + \lambda \frac{2}{N}  ||| \res{\bHhat}{\supp^c\supp}  \res{\bHhat}{\supp\supp}^{\minus 1} |||_{\infty} \| \bz_{\supp} \|_{\infty}  + \frac{2}{N} \| \sum_{x \in D} \bpsi_{\supp^c}(x) \sum_{j \ne i} \sum_{k=r+1}^\infty \beta_{ijk}^* \widetilde{\psi}_k(x_i, x_j) \|_{\infty}+\\
&\frac{2}{N} \| \sum_{x \in D} \bpsi_{\supp^c}(x) \bgamma_{\supp}^\T(x) \bbeta_{\supp} \|_{\infty} 
\end{split}
\end{align}
Using Assumption \ref{assum:mutual incoherence} and noting that $\| \bz_{\supp} \|_{\infty} \leq 1$, we have
\begin{align}
\label{eq:normz}
\begin{split}
&\lambda \| \bz_{\supp^c} \|_{\infty}  \leq  (1 - \frac{\alpha}{2}) \frac{2}{N} \| (\sum_{y \in D} \bpsi_{\supp}(y) \sum_{j \ne i} \sum_{k=r+1}^\infty \beta_{ijk}^*   \widetilde{\psi}(y_i, y_j)) \|_{\infty} + (1 - \frac{\alpha}{2}) \frac{2}{N} \| (\sum_{y \in D} \bpsi_{\supp}(y) \bgamma_{\supp}^\T(y)  \bbeta_{\supp}) \|_{\infty} \\
&+ \lambda \frac{2}{N} (1 - \frac{\alpha}{2})   + \frac{2}{N} \| \sum_{x \in D} \bpsi_{\supp^c}(x) \sum_{j \ne i} \sum_{k=r+1}^\infty \beta_{ijk}^* \widetilde{\psi}_k(x_i, x_j) \|_{\infty}+\frac{2}{N} \| \sum_{x \in D} \bpsi_{\supp^c}(x) \bgamma_{\supp}^\T(x) \bbeta_{\supp} \|_{\infty} 
\end{split}
\end{align}

We want to show that $\| \bz_{\supp^c} \|_{\infty} < 1$, which ensures that $\beta_{ijk} = 0, \forall j \notin \supp$. We do this by bounding each term in Equation \eqref{eq:normz} using the following lemmas.
\begin{lemma}
	\label{lem:psi gamma beta}
	For some $\epsilon > 0$, we have
	\begin{align}
	\begin{split}
	&\prob(\frac{1}{N} \| \sum_{x \in D} \bpsi_{\supp^c}(x) \bgamma_{\supp}^\T(x) \bbeta_{\supp}  \|_{\infty} \geq \epsilon) \leq 2 (n-d) \exp( \frac{- N \epsilon^2}{ 2  \sigma^2 \overline\bpsi^2 C^2 })
	\end{split}
	\end{align}
	as well as
	\begin{align}
	\begin{split}
	&\prob(\frac{1}{N} \| \sum_{y \in D} \bpsi_{\supp}(y) \bgamma_{\supp}^\T(y) \bbeta_{\supp}  \|_{\infty} \geq \epsilon) \leq 2 d \exp(\frac{ -N \epsilon^2}{ 2  \sigma^2 \overline\bpsi^2 C^2 })
	\end{split}
	\end{align}
\end{lemma} 
(Check Appendix \ref{proof:psi gamma beta} for details.)

\begin{lemma}
	\label{lem:psi beta tildepsi}
	For sufficiently large $r$ and $\epsilon > 0$, we have
	\begin{align}
		\begin{split}
		&\prob(\frac{1}{N}\| (\sum_{y \in D} \bpsi_{\supp}(y) \sum_{j \ne i} \sum_{k=r+1}^\infty \beta_{ijk}^* \widetilde{\psi}(y_i, y_j)) \|_{\infty} \geq 2 \epsilon)  \leq 2 d \exp(- \frac{N \epsilon^2}{2 \overline{\bpsi}^2 d \delta }  )
		\end{split}
	\end{align}
	as well as
	\begin{align}
		\begin{split}
		&\prob(\frac{1}{N}\| (\sum_{y \in D} \bpsi_{\supp^c}(y) \sum_{j \ne i} \sum_{k=r+1}^\infty \beta_{ijk}^* \widetilde{\psi}(y_i, y_j)) \|_{\infty} \geq 2 \epsilon) \leq 2 (n-d) \exp(- \frac{N \epsilon^2}{2 \overline{\bpsi}^2 d \delta }  )
		\end{split}
	\end{align}
	where $\delta > 0$ is an arbitrary constant which depends on $r$.	
\end{lemma}
(Check Appendix \ref{proof:psi beta tildepsi} for details.)

Combining all the results together,
\begin{align}
\begin{split}
\lambda \| \bz_{\supp^c} \|_{\infty}  \leq&  (1 - \frac{\alpha}{2}) 4 \epsilon_1 + 
(1 - \frac{\alpha}{2}) 2 \epsilon_2+ \lambda \frac{2}{N} (1 - \frac{\alpha}{2}) +4 \epsilon_3 + 2 \epsilon_4 
\end{split}
\end{align}

We take $\epsilon_1 \leq \frac{\lambda}{20}, \epsilon_2 \leq \frac{\lambda}{10}, \epsilon_3 \leq (1 - \frac{\alpha}{2})\frac{\lambda}{20}$ and $\epsilon_4 \leq (1 - \frac{\alpha}{2}) \frac{\lambda}{10}$. Then for $N \geq 10$, we have
\begin{align}
\begin{split}
\lambda \| \bz_{\supp^c} \|_{\infty}  &\leq  (1 - \frac{\alpha}{2}) \lambda \\
\| \bz_{\supp^c} \|_{\infty}  &\leq 1 - \frac{\alpha}{2}
\end{split}
\end{align}

which justifies our choice that $\beta_{ijk} = 0, \forall j \ne \supp$. The choice of $\epsilon_i, \forall i \in \{1,\dots,4\}$ imposes restrictions on choice of $\lambda$ for the concentration bounds to hold with high probability. In particular, we need 
\begin{align}
\begin{split}
&\lambda \geq \max\big( 20 \sqrt{2} \overline\psi \sqrt{\frac{d\delta\log d}{N}}, 10 \sqrt{2} \sigma \overline\psi C\sqrt{\frac{\log d}{N}},\frac{20 \sqrt{2}\overline\psi }{1 - \frac{\alpha}{2}} \sqrt{\frac{d\delta\log (n-d)}{N}}, \frac{10 \sqrt{2} \sigma \overline\psi C }{1 - \frac{\alpha}{2}} \sqrt{\frac{\log (n-d)}{N}}  \big)
\end{split}
\end{align}

\subsection{Estimation Error}
\label{subsec:estimation error}

Now that we have a criteria for choosing the regularization parameter $\lambda$, we can put an upper bound on the estimation error between $\bbeta$ and $\bbeta^*$. We know from equation \eqref{eq:delbeta} that, 
\begin{align*}
	\begin{split}
	&\| \bbeta_{\supp}^* - \bbeta_{\supp} \|_2 = \| - \res{\bHhat}{\supp\supp}^{\minus 1} \frac{1}{N} (\sum_{x \in D} \bpsi_{\supp}(x) \sum_{j \ne i} \sum_{k=r+1}^\infty \beta_{ijk}^* \widetilde{\psi}_k(x_i, x_j)) -\res{\bHhat}{\supp\supp}^{\minus 1} (\frac{1}{N}\sum_{x \in D} \bpsi_{\supp}(x) \bgamma_{\supp}^\T(x) \bbeta_{\supp}) +\\
	& \lambda \res{\bHhat}{\supp\supp}^{\minus 1} \frac{1}{N} \bz_{\supp} \|_2
	\end{split}
\end{align*}
Using triangle norm inequality and noticing that $\| A x\|_2 \leq \| A \|_2 \| x \|_2$, where $\| A \|_2$ is the spectral norm of matrix $A$, we can write
\begin{align*}
\begin{split}
	&\| \bbeta_{\supp}^* - \bbeta_{\supp} \|_2 \leq \| \res{\bHhat}{\supp\supp}^{\minus 1} \|_2 \|\frac{1}{N} (\sum_{x \in D} \bpsi_{\supp}(x) \sum_{j \ne i} \sum_{k=r+1}^\infty \beta_{ijk}^* \widetilde{\psi}_k(x_i, x_j)) \|_2 + \| \res{\bHhat}{\supp\supp}^{\minus 1} \|_2 \| \frac{1}{N} (\sum_{x \in D} \bpsi_{\supp}(x) \bgamma_{\supp}^\T(x) \bbeta_{\supp}\\
	& ) \|_2 +\lambda \| \res{\bHhat}{\supp\supp}^{\minus 1} \|_2 \|  \frac{1}{N} \bz_{\supp} \|_2
\end{split}
\end{align*}
We have already shown that $ \Lambda_{\min}(\res{\bHhat}{\supp\supp}) \geq \frac{C_{\min}}{2}$. It follows that $ \Lambda_{\min}(\res{\bHhat}{\supp\supp}^{\minus 1}) \leq \frac{2}{C_{\min}}$. We also note that $\| \bz_{\supp} \|_2 \leq \sqrt{|\supp |} \| \bz_{\supp} \|_{\infty} \leq \sqrt{|\supp |}$. Thus,
\begin{align}
\label{eq:delbeta simpler}
\begin{split}
	&\| \bbeta_{\supp}^* - \bbeta_{\supp} \|_2 \leq \frac{2}{C_{\min}} (\|\frac{1}{N} (\sum_{x \in D} \bpsi_{\supp}(x) \sum_{j \ne i} \sum_{k=r+1}^\infty \beta_{ijk}^* \widetilde{\psi}_k(x_i, x_j)) \|_2 +  \|\frac{1}{N} (\sum_{x \in D} \bpsi_{\supp}(x) \bgamma_{\supp}^\T(x) \bbeta_{\supp}) \|_2 +\lambda \frac{1} {\sqrt{|\supp |}} )
\end{split}
\end{align}
It remains to show that $\|\frac{1}{N} (\sum_{x \in D} \bpsi_{\supp}(x) \sum_{j \ne i} \sum_{k=r+1}^\infty \beta_{ijk}^* \widetilde{\psi}_k(x_i, x_j)) \|_2$ and $\|\frac{1}{N} (\sum_{x \in D} \bpsi_{\supp}(x) \bgamma_{\supp}^\T(x) \bbeta_{\supp}) \|_2$ are bounded which we will do in the following lemma.

\begin{lemma}
	\label{lem:l2 norm bounds}
	For some $\epsilon > 0$,
	\begin{align}
	\begin{split}
		&\prob(\|\frac{1}{N} (\sum_{x \in D} \bpsi_{\supp}(x) \bgamma_{\supp}^\T(x) \bbeta_{\supp}) \|_2 \geq \epsilon) \leq 2 d \exp( \frac{- N \epsilon^2}{2 \sigma^2 d\widebar\psi^2C^2})
	\end{split}
	\end{align}
	and
	\begin{align}
		\begin{split}
		&\prob(\|\frac{1}{N} (\sum_{x \in D} \bpsi_{\supp}(x) \sum_{j \ne i} \sum_{k=r+1}^\infty \beta_{ijk}^* \widetilde{\psi}_k(x_i, x_j)) \|_2 \geq \epsilon) \leq 2 d \exp( \frac{- N \epsilon^2}{2 d\widebar\psi^2 d \delta})
		\end{split}
	\end{align}
	where $\delta > 0$ could be an arbitrary constant which depends on $r$.
\end{lemma}
(Check Appendix \ref{proof:l2 norm bounds} for details.)

Using results from Lemma \ref{lem:l2 norm bounds}, we can rewrite Equation \eqref{eq:delbeta simpler} as,
\begin{align}
\begin{split}
	&\| \bbeta_{\supp}^* - \bbeta_{\supp} \|_2 \leq \frac{2}{C_{\min}} (\epsilon_1 +  \epsilon_2 +\lambda \frac{1} {\sqrt{|\supp |}} )
\end{split}
\end{align}
Substituting $\epsilon_1 = \sqrt{d} \lambda$ and  $\epsilon_2 = \sqrt{d} \lambda$ we get,
\begin{align}
\begin{split}
&\| \bbeta_{\supp}^* - \bbeta_{\supp} \|_2 \leq \frac{2 \lambda}{C_{\min}} ( 2\sqrt{d} +\frac{1} {\sqrt{|\supp |}} )
\end{split}
\end{align}
Note that $\| \bbeta_{\supp}^* - \bbeta_{\supp} \|_{\infty} \leq \| \bbeta_{\supp}^* - \bbeta_{\supp} \|_2$, thus
\begin{align}
\label{eq:delbeta final}
\begin{split}
&\| \bbeta_{\supp}^* - \bbeta_{\supp} \|_{\infty} \leq \frac{2 \lambda}{C_{\min}} ( 2\sqrt{d} +\frac{1} {\sqrt{|\supp |}} )
\end{split}
\end{align}
Equation \eqref{eq:delbeta final} provides a \emph{minimum weight criteria}, i.e., if $\min_{j \in \supp} |\beta_j^*| \geq \Delta(\lambda, C_{\min}, d) = \frac{2 \lambda}{C_{\min}} ( 2\sqrt{d} +\frac{1} {\sqrt{|\supp |}} ) $ then we recover $\bbeta^*$ up to correct sign. This ensures that we recover the exact structure of the true game.

\subsection{Recovering Payoffs.}
\label{subsec:recovering payoffs}
In this subsection, we show that the recovered payoff function $\widehat{u}_i(x)$ is not far away from the true payoff function $u_i^*(x)$. To that end we provide a bound between the recovered payoff function and the true payoff function. 
\begin{align*}
\begin{split}
	| u_i^*(x) - \widehat{u}_i(x) | &= | \sum_{j \in \supp} \sum_{k=0}^{\infty} \beta_{ijk}^* \psi_k(x_i, x_j) -\sum_{j \in \supp} \sum_{k=0}^r \beta_{ijk} \psi_k(x_i, x_j)  | \\
	&= |  \sum_{j \in \supp} \sum_{k=0}^r (\beta_{ijk}^* - \beta_{ijk}) \psi_k(x_i, x_j) +  \sum_{j \in \supp} \sum_{k=r+1}^{\infty} \beta_{ijk}^*\psi_k(x_i, x_j)  | \\
	&\leq \| \bbeta^*_{\supp} - \bbeta_{\supp} \|_2 \| \bpsi_{\supp} \|_2 + \widebar{\psi} \delta\\
	&\leq \frac{2 \lambda}{C_{\min}} ( 2\sqrt{d} +\frac{1} {\sqrt{|\supp |}} ) \sqrt{d} \widebar{\psi} + \widebar{\psi} \delta \\
	&= \epsilon(\lambda, C_{\min}, d, \widebar{\psi}, \delta)
\end{split}
\end{align*} 
It follows that the recovered payoff function is $\epsilon(\lambda, C_{\min}, d, \widebar{\psi}, \delta)$ away from the true payoff function. Furthermore, their distance decreases as we choose smaller $\lambda$ and smaller $\delta$. This can be done by increasing the number of samples and the parameter $r$ respectively.   

\section{Conclusion}
\label{sec:conclusion}

In this paper, we have proposed an $l_1$-regularized method which recovers the exact structure of true game under minimum weight criteria by learning a finite dimensional vector. We also showed that we recover a payoff function which is $\epsilon$ close to the true payoff function. All of our high probability statements hold as long as $N = \mathcal{O}(\frac{\widebar{\psi}^2 \sigma^2 C^2 \delta r^3d^3}{\epsilon^2}\log (rn))$. Furthermore, our method solves a convex optimization problem and thus it runs in polynomial time.

\onecolumn
\appendix
\paragraph{Note: } Please see Appendix \ref{sec:real world data} for experimental results on real-world data.
\section{Proof of Lemma \ref{lem:kkt}}
\label{proof:kkt}
\paragraph{Lemma \ref{lem:kkt}} [KKT conditions]
	\emph{The following Karush-Kuhn-Tucker (KKT) conditions hold at the optimal solution of the optimization problem \eqref{eq:estimateprob}:
	\begin{align}
	\begin{split}
	\text{Stationarity: } &\frac{\partial}{\partial \bbeta} \big[ \frac{1}{N} \sum_{x \in D} (\widetilde{u}(x) - \overline{u}(x))^2 \big] + \lambda \bz = 0 \\
	\text{Primal Feasibility: } & w \leq C \\
	& w =\sum_{j \ne i}  \sum_{k=0}^r |\beta_{ijk}|  \\
	\text{Dual Feasibility: } & \lambda \geq 0
	\end{split}
	\end{align} 
	where $\bz \in \real^{r(n-1) \times 1} $indexed by $(j,k), j \in \seq{n}, j \ne i, k \in \seq{r}$. It is defined as follows:
	\begin{align}
	z_{ijk} = \begin{cases}
	\sign(\beta_{ijk}), \text{ if } \beta_{ijk} \ne 0 \\
	[-1, 1],  \text{ otherwise }
	\end{cases}
	\end{align}	
}
\begin{proof}
Equivalently, 
\begin{align*}
\begin{split}
&\min_{\beta} \frac{1}{N} \sum_{x \in D} (\widetilde{u}(x) - \overline{u}(x))^2 + \lambda w \\
&\text{such that } w  \leq C \\
&\quad\quad\quad  w = \sum_{j\ne i}\sum_{k=0}^r |\beta_{ijk}| 
\end{split}
\end{align*}
The Lagrangian for the above can be written as,
\begin{align}
\label{eq:lagrangean}
L(\bbeta, w; \mu, \eta) = \frac{1}{N} \sum_{x \in D} (\widetilde{u}(x) - \overline{u}(x))^2 + \lambda w + \mu (w  - C) + \eta (w - \sum_{j\ne i}\sum_{k=0}^r |\beta_{ijk}|)
\end{align}
where $\mu \geq 0$. Alternatively,
\begin{align}
\label{eq:lagrangean1}
L(\bbeta, w; \mu, \eta) = \frac{1}{N} \sum_{x \in D} (\widetilde{u}(x) - \overline{u}(x))^2 + \lambda w + \mu (w  - C) + \eta (w - \mb{z}^\T \bbeta)
\end{align}

Here, $\bbeta \in \real^{r(n-1) \times 1}, \bz \in \real^{r(n-1) \times 1} $ with all of them indexed by $(j,k), j \in \seq{n}, j \ne i, k \in \seq{r}$. Vector $\bz \in \real^{r(n-1) \times 1}$ is defined as follows:
\begin{align}
z_{ijk} = \begin{cases}
\sign(\beta_{ijk}), \text{ if } \beta_{ijk} \ne 0 \\
[-1, 1],  \text{ otherwise }
\end{cases}
\end{align}	
By writing the KKT conditions from the Lagrangian in Equation \eqref{eq:lagrangean1}, a solution is optimal if and only if: 
\begin{align}
\label{eq:kkt}
\begin{split}
\text{Stationarity: } &\frac{\partial}{\partial \bbeta} \big[ \frac{1}{N} \sum_{x \in D} (\widetilde{u}(x) - \overline{u}(x))^2 \big] - \eta \bz = 0 \\
& \lambda + \mu + \eta = 0 \\
\text{Complimentarity: } &  \mu (w - C) = 0 \\
\text{Primal Feasibility: } & w \leq C \\
& w = \sum_{k=0}^r |\beta_{ijk}|  \\
\text{Dual Feasibility: } & \mu \geq 0
\end{split}
\end{align} 
If we choose $\mu = 0$, then the following KKT conditions must hold at the optimal solution:
\begin{align}
\begin{split}
\text{Stationarity: } &\frac{\partial}{\partial \bbeta} \big[ \frac{1}{N} \sum_{x \in D} (\widetilde{u}(x) - \overline{u}(x))^2 \big] + \lambda \bz = 0 \\
\text{Primal Feasibility: } & w \leq C \\
& w = \sum_{j\ne i}\sum_{k=0}^r |\beta_{ijk}|  
\end{split}
\end{align}
\end{proof}

\section{Proof of Lemma \ref{lem:positive definiteness}}
\label{proof:positive definiteness}
\paragraph{Lemma \ref{lem:positive definiteness}}
	\emph{If $\Lambda_{\min} (\res{\bH}{\supp\supp}) = C_{\min} > 0$ then  $\Lambda_{\min} (\res{\bHhat}{\supp\supp}) = C_{\min} - \epsilon$ for some $\epsilon > 0$ with probability at least $\exp(\frac{-N\epsilon^2}{2 \budget^2 \sigma^2} + \mathcal{O}(rd) )$.
} 
\begin{proof}
	Let $\hat{C}_{\min}$ is minimum eigenvalue of $\bHhat_{\supp\supp}$. Then,
	\begin{align}
	\begin{split}
	\hat{C}_{\min} &= \min_{y, \| y \|_2 = 1} y^\T \bHhat_{\supp\supp} y \\
	&= \min_{y, \| y \|_2 = 1} y^\T \bH y + y^\T \frac{1}{N} \sum_{x \in D} [\psi(x) \gamma(x)^\T]_{\supp\supp} y \\
	&\geq C_{\min} + \min_{y, \| y \|_2 = 1} y^\T \frac{1}{N} \sum_{x \in D} [\psi(x) \gamma(x)^\T]_{\supp\supp} y
	\end{split}
	\end{align}
	For some $y$ such that $\| y \|_2 = 1$, we have a random variable $R \triangleq y^\T \frac{1}{N} \sum_{x \in D} [\psi(x) \gamma(x)^\T]_{\supp\supp} y$, then $R$ is a zero mean sub-Gaussian random variable with variance proxy $\frac{\sum_{s=1}^{N} a_s^2 \sigma^2}{N^2}$. Then using the tail bound for sub-Gaussian random variables, we can write
	
	\begin{align}
	\begin{split}
	\prob(R \leq -\epsilon) \leq \exp(\frac{-\epsilon^2}{2 \frac{\sum_{s=1}^{N} a_s^2 \sigma^2}{N^2}})
	\end{split}
	\end{align}
	
	We assume that $\| \psi \|_2 \leq \budget$, then $ \max_{y, \|y\|_2 = 1} a_s \leq \budget $. Thus,
	\begin{align}
	\begin{split}
	\prob(R \leq -\epsilon) \leq \exp(\frac{-N\epsilon^2}{2 \budget^2 \sigma^2})
	\end{split}
	\end{align}	
	
	Using the $\epsilon$-nets argument and taking a union bound across $\exp(\mathcal{O}(rd))$ $y$ of the net, we can write
	\begin{align}
	\begin{split}
	\prob(R \leq -\epsilon) \leq \exp(\frac{-N\epsilon^2}{2 \budget^2 \sigma^2} + \mathcal{O}(rd) )
	\end{split}
	\end{align}	
	Clearly, if $N = \mathcal{O}( \frac{8\budget^2\sigma^2}{C_{\min}^2} rd)$ we have
	$\hat{C}_{\min} \geq \frac{C_{\min}}{2}$ with high probability.
\end{proof}

\section{Proof of Lemma \ref{lem:mutual incoherence}}
\label{proof:mutual incoherence}
\paragraph{Lemma \ref{lem:mutual incoherence}}
\emph{If $\|| \res{\bH}{\supp^c\supp} \res{\bH}{\supp\supp}^{-1}\||_{\infty} \leq 1 - \alpha$ for some $0 < \alpha < 1$ then $\|| \res{\bHhat}{\supp^c\supp} \res{\bHhat}{\supp\supp}^{-1}\||_{\infty} \leq 1 - \frac{\alpha}{2}$ with probability at least $1 - \exp(\frac{-K NC_{\min}^2 \alpha^2}{(1-\alpha)^2 r^3 d^3 \widebar\psi^2\sigma^2} + \log(2r^2(n-d)d)) - \exp(\frac{-N \alpha}{6 r^2 d^2 \widebar\psi^2\sigma^2} + \log(2r^2(n-d)d)) - \exp(\frac{-2N}{\budget^2\sigma^2C_{\min}^2} + \mathcal{O}(rd)) - 2 \exp(\frac{-N\alpha^2}{24r^3d^3\widebar\psi^2\sigma^2} + \log(r^2d^2)) $ for a constant $K>0$.
}
\begin{proof}
	We start the proof by first proving one auxiliary lemma. 
	\begin{lemma}
		For any $\delta > 0$, the following holds:
		\begin{align}
		\label{eq:hscs}
		\begin{split}
		\prob(\|| \res{\bHhat}{\supp^c\supp} - \res{\bH}{\supp^c\supp} \||_{\infty} \geq \delta) \leq \exp( \frac{-N\epsilon^2}{ r^2 d^2 {\overline{\psi}}^2 \sigma^2} + \log (2 r^2 (n-d) d) )
		\end{split}
		\end{align}
		\begin{align}
		\label{eq:hss}
		\begin{split}
		\prob(\|| \res{\bHhat}{\supp\supp} - \res{\bH}{\supp\supp} \||_{\infty} \geq \delta) \leq \exp( -\frac{-N\epsilon^2}{ r^2 d^2 {\overline{\psi}}^2 \sigma^2} + \log (2 r^2 d^2 ) )
		\end{split}
		\end{align}
		\begin{align}
		\label{eq:invhss}
		\begin{split}
		\prob(\|| \res{\bHhat}{\supp\supp}^{-1} - \res{\bH}{\supp\supp}^{-1} \||_{\infty} \geq \delta) \leq  \exp(\frac{-2N}{ \budget^2 \sigma^2 C_{\min}^2} + \mathcal{O}(rd) ) +  2 \exp(\frac{ -N \delta^2 }{4 r^3 d^3 \overline{\psi}^2 \sigma^2} + \log (r^2 d^2))
		\end{split}
		\end{align}
	\end{lemma}
	\begin{proof}
		Note that,
		\begin{align}
		\begin{split}
		[\res{\bHhat}{\supp^c\supp} - \res{\bH}{\supp^c\supp} ]_{jk} = [\frac{1}{N} \sum_{x \in D} \res{\psi(x) \gamma(x)^\T}{\supp^c\supp}]_{jk}
		\end{split}
		\end{align}
		We further note that the random variable $[\frac{1}{N} \sum_{x \in D} \res{\psi(x) \gamma(x)^\T}{\supp^c\supp}]_{jk}$ is a zero mean sub-Gaussian random variable with variance proxy $\frac{\sigma^2 \sum_{x \in D} {\psi(x)_j^2}}{N^2}$. Thus,
		\begin{align}
		\begin{split}
		\prob(|[\frac{1}{N} \sum_{x \in D} \res{\psi(x) \gamma(x)^\T}{\supp^c\supp}]_{jk} | \geq \epsilon) \leq 2 \exp(\frac{-\epsilon^2}{\frac{\sigma^2 \sum_{x \in D} {\psi(x)_j^2}}{N^2} }) \leq 2 \exp( -\frac{-N\epsilon^2}{{\overline{\psi}}^2 \sigma^2}  )
		\end{split}
		\end{align}
		Taking $\epsilon = \frac{\epsilon}{rd}$, we have
		\begin{align}
		\begin{split}
		\prob(|[\frac{1}{N} \sum_{x \in D} \res{\psi(x) \gamma(x)^\T}{\supp^c\supp}]_{jk} | \geq \frac{\epsilon}{rd}) \leq 2 \exp( -\frac{-N\epsilon^2}{ r^2 d^2 {\overline{\psi}}^2 \sigma^2}  )
		\end{split}
		\end{align}
		
		We observe that,
		\begin{align}
		\begin{split}
		\|| \res{\bHhat}{\supp^c\supp} - \res{\bH}{\supp^c\supp} \||_{\infty} = \max_{j \in \supp^c} \sum_{k \in \supp} |[\frac{1}{N} \sum_{x \in D} \res{\psi(x) \gamma(x)^\T}{\supp^c\supp}]_{jk}
		\end{split}
		\end{align}
		Thus, taking a union bound across $j \in \supp^c$ and $k \in \supp$, we get
		\begin{align}
		\begin{split}
		\prob(\|| \res{\bHhat}{\supp^c\supp} - \res{\bH}{\supp^c\supp} \||_{\infty}  \geq \epsilon) &\leq 2 r^2 (n-d) d \exp( -\frac{-N\epsilon^2}{ r^2 d^2 {\overline{\psi}}^2 \sigma^2}  ) \\
		&= \exp( \frac{-N\epsilon^2}{ r^2 d^2 {\overline{\psi}}^2 \sigma^2} + \log (2 r^2 (n-d) d) )
		\end{split}
		\end{align}
		Similarly,
		\begin{align}
		\begin{split}
		\|| \res{\bHhat}{\supp\supp} - \res{\bH}{\supp\supp} \||_{\infty} = \max_{j \in \supp} \sum_{k \in \supp} |[\frac{1}{N} \sum_{x \in D} \res{\psi(x) \gamma(x)^\T}{\supp\supp}]_{jk}
		\end{split}
		\end{align}
		It follows that,
		\begin{align}
		\begin{split}
		\prob(\|| \res{\bHhat}{\supp\supp} - \res{\bH}{\supp\supp} \||_{\infty}  \geq \epsilon) &\leq 2 r^2 d^2 \exp( \frac{-N\epsilon^2}{ r^2 d^2 {\overline{\psi}}^2 \sigma^2}  ) \\
		&= \exp( -\frac{-N\epsilon^2}{ r^2 d^2 {\overline{\psi}}^2 \sigma^2} + \log (2 r^2 d^2 ) )
		\end{split}
		\end{align}
		Now,
		\begin{align}
		\begin{split}
		\|| \res{\bHhat}{\supp\supp}^{-1} - \res{\bH}{\supp\supp}^{-1} \||_{\infty} &= \|| \res{\bH}{\supp\supp}^{-1} [ \res{\bH}{\supp\supp} - \res{\bHhat}{\supp\supp}] \res{\bHhat}{\supp\supp}^{-1} \||_{\infty} \\
		&\leq \sqrt{dr} \|| \res{\bH}{\supp\supp}^{-1} [ \res{\bH}{\supp\supp} - \res{\bHhat}{\supp\supp}] \res{\bHhat}{\supp\supp}^{-1} \||_2 \\
		&\leq \sqrt{dr} \|| \res{\bH}{\supp\supp}^{-1} \||_2 \|| [ \res{\bH}{\supp\supp} - \res{\bHhat}{\supp\supp}] \||_2 \||  \res{\bHhat}{\supp\supp}^{-1} \||_2 \\
		&\leq \frac{\sqrt{dr} }{C_{\min}} \|| [ \res{\bH}{\supp\supp} - \res{\bHhat}{\supp\supp}] \||_2 \||  \res{\bHhat}{\supp\supp}^{-1} \||_2
		\end{split}
		\end{align}
		We have shown that,
		\begin{align}
		\begin{split}
		\prob(\||  \res{\bHhat}{\supp\supp}^{-1} \||_2 \geq \frac{2}{C_{\min}} ) \leq \exp(\frac{-\epsilon^2}{2 \budget^2 \sigma^2} + \mathcal{O}(rd) )
		\end{split}
		\end{align} 
		Moreover,
		\begin{align}
		\begin{split}
		\prob( \|| [ \res{\bH}{\supp\supp} - \res{\bHhat}{\supp\supp}] \||_2 \geq \epsilon) \leq 2 \exp(\frac{ -N \epsilon^2 }{ r^2 d^2 \overline{\psi}^2 \sigma^2} + \log (r^2 d^2))
		\end{split}
		\end{align}
		Taking $\epsilon = \frac{\delta}{2\sqrt{d r}}$, we have
		\begin{align}
		\begin{split}
		\prob( \|| [ \res{\bH}{\supp\supp} - \res{\bHhat}{\supp\supp}] \||_2 \geq \frac{\delta}{2\sqrt{d r}}) \leq 2 \exp(\frac{ -N \delta^2 }{4 r^3 d^3 \overline{\psi}^2 \sigma^2} + \log (r^2 d^2))
		\end{split}
		\end{align}
		and furthermore
		\begin{align}
		\begin{split}
		\prob(\|| \res{\bHhat}{\supp\supp}^{-1} - \res{\bH}{\supp\supp}^{-1} \||_{\infty} \geq \delta ) \leq \exp(\frac{-2N}{ \budget^2 \sigma^2 C_{\min}^2} + \mathcal{O}(rd) ) +  2 \exp(\frac{ -N \delta^2 }{4 r^3 d^3 \overline{\psi}^2 \sigma^2} + \log (r^2 d^2))
		\end{split}
		\end{align} 
	\end{proof}
	Now we are ready to prove the main lemma.
	\begin{align}
	\begin{split}
	\res{\bHhat}{\supp^c\supp} \res{\bHhat}{\supp\supp}^{-1} = T_1 + T_2 + T_3 + T_4
	\end{split}
	\end{align}
	where
	\begin{align}
	\begin{split}
	T_1 \triangleq \res{\bH}{\supp^c\supp}[\res{\bHhat}{\supp\supp}^{-1} - \res{\bH}{\supp\supp}^{-1}], \\
	T_2 \triangleq [\res{\bHhat}{\supp^c\supp} - \res{\bH}{\supp^c\supp}] \res{\bH}{\supp\supp}^{-1} \\
	T_3 \triangleq [\res{\bHhat}{\supp^c\supp} - \res{\bH}{\supp^c\supp}] [\res{\bHhat}{\supp\supp}^{-1} - \res{\bH}{\supp\supp}^{-1}], \\
	T_4 \triangleq \res{\bH}{\supp^c\supp}\res{\bH}{\supp\supp}^{-1}
	\end{split}
	\end{align}
	We know that $\|| T_4 \||_{\infty} \leq 1 - \alpha$. Now,
	\begin{align}
	\begin{split}
	\|| T_1 \||_{\infty} &= \|| \res{\bH}{\supp^c\supp}[\res{\bHhat}{\supp\supp}^{-1} - \res{\bH}{\supp\supp}^{-1}] \||_{\infty} \\
	&\leq \|| \res{\bH}{\supp^c\supp} \res{\bH}{\supp\supp}^{-1} \||_{\infty} \|| \res{\bHhat}{\supp^c\supp} - \res{\bH}{\supp^c\supp} \||_{\infty} \|| \res{\bHhat}{\supp\supp}^{-1}\||_{\infty} \\
	&\leq (1 - \alpha) \|| \res{\bHhat}{\supp^c\supp} - \res{\bH}{\supp^c\supp} \||_{\infty} \sqrt{rd} \| \res{\bHhat}{\supp\supp}\|_2 \\
	&\leq (1 - \alpha) \|| \res{\bHhat}{\supp^c\supp} - \res{\bH}{\supp^c\supp} \||_{\infty} \sqrt{rd} \frac{2}{C_{\min}} \\
	&\leq \frac{\alpha}{6}
	\end{split}
	\end{align}
	The last step holds with probability at least $1 - \exp(\frac{-NC_{\min}^2 \alpha^2}{144(1-\alpha)^2 r^3 d^3 \widebar\psi^2\sigma^2} + \log(2r^2(n-d)d)) $ by choosing $\delta = \frac{C_{\min} \alpha}{12(1 - \alpha) \sqrt{rd}}$ in equation \eqref{eq:hscs}. 			
	
	For the second term,
	\begin{align}
	\begin{split}
	\||T_2 \||_{\infty} &\leq \sqrt{rd} \|| \res{\bH}{\supp\supp}^{-1} \||_2 \||[\res{\bHhat}{\supp^c\supp} - \res{\bH}{\supp^c\supp}]\||_{\infty} \\
	&\leq \frac{\sqrt{rd}}{C_{\min}} \||[\res{\bHhat}{\supp^c\supp} - \res{\bH}{\supp^c\supp}]\||_{\infty} \\
	&\leq \frac{\alpha}{6}
	\end{split}
	\end{align}
	The last step holds with probability at least $1 - \exp(\frac{-NC_{\min}^2 \alpha^2}{36 r^3 d^3 \widebar\psi^2\sigma^2} + \log(2r^2(n-d)d)) $ by choosing $\delta =\frac{C_{\min} \alpha}{6 \sqrt{rd}} $.
	
	For the third term,
	\begin{align}
	\begin{split}
	\|| T_3 \||_{\infty} &\leq \|| [\res{\bHhat}{\supp^c\supp} - \res{\bH}{\supp^c\supp}] \||_{\infty} \|| [\res{\bHhat}{\supp\supp}^{-1} - \res{\bH}{\supp\supp}^{-1}] \||_{\infty} \\
	&\leq \frac{\alpha}{6}
	\end{split}
	\end{align} 
	The last step holds with probability at least $1 - \exp(\frac{-N \alpha}{6 r^2 d^2 \widebar\psi^2\sigma^2} + \log(2r^2(n-d)d)) - \exp(\frac{-2N}{\budget^2\sigma^2C_{\min}^2} + \mathcal{O}(rd)) - 2 \exp(\frac{-N\alpha^2}{24r^3d^3\widebar\psi^2\sigma^2} + \log(r^2d^2)) $ by choosing $\delta = \sqrt{\frac{\alpha}{6}}$ in equations \eqref{eq:hscs} and \eqref{eq:invhss}.
	
	Combining all the terms together, we get the final result. 
\end{proof}

\section{Proof of Lemma \ref{lem:psi gamma beta}}
\label{proof:psi gamma beta}
\paragraph{Lemma \ref{lem:psi gamma beta}}
\emph{For some $\epsilon > 0$, we have
	\begin{align}
	\begin{split}
	&\prob(\frac{1}{N} \| \sum_{x \in D} \bpsi_{\supp^c}(x) \bgamma_{\supp}^\T(x) \bbeta_{\supp}  \|_{\infty} \geq \epsilon) \leq 2 (n-d) \exp( \frac{- N \epsilon^2}{ 2  \sigma^2 \overline\bpsi^2 C^2 })
	\end{split}
	\end{align}
	as well as
	\begin{align}
	\begin{split}
	&\prob(\frac{1}{N} \| \sum_{y \in D} \bpsi_{\supp}(y) \bgamma_{\supp}^\T(y) \bbeta_{\supp}  \|_{\infty} \geq \epsilon) \leq 2 d \exp(\frac{ -N \epsilon^2}{ 2  \sigma^2 \overline\bpsi^2 C^2 })
	\end{split}
	\end{align}
}
\begin{proof}
	Let $R(x) \triangleq  \bpsi_{\supp^c}(x) \bgamma_{\supp}^\T(x) \bbeta_{\supp}$ be a random variable and $R_j(x)$ be the $j$th entry of $R(x)$. Then $R_j(x) = \bpsi_{\supp^c}(x)_j \bgamma_{\supp}^\T(x) \bbeta_{\supp} $. Now note that $\bgamma_{\supp}(x)$ is a zero mean sub-Gaussian random variable with variance proxy $\sigma^2 \mb{I}_{d \times d}$. Thus, $\bgamma_{\supp}^\T(x) \bbeta_{\supp}$ is also a zero mean sub-Gaussian random variable with variance proxy $\sigma^2 \| \bbeta_{\supp} \|_2^2$ and $R_j(x)$ is a zero mean sub-Gaussian random variable with variance proxy $\sigma^2 \bpsi_{\supp^c}(x)_j^2 \| \bbeta_{\supp} \|_2^2$. For a given $\bbeta$, using a concentration bound for sub-Gaussian random variables, we can write the following:
	\begin{align}
	\begin{split}
	\prob(|R_j(x)| \geq \epsilon | \bbeta) \leq 2 \exp(- \frac{ \epsilon^2}{ 2  \sigma^2 \bpsi_{\supp^c}(x)_j^2 \| \bbeta_{\supp} \|_2^2 })
	\end{split}
	\end{align}
	Note that from Equation \eqref{eq:kkt}, $\| \bbeta \|_1 \leq C \implies \| \bbeta \|_2^2 \leq C^2$. Thus, 
	\begin{align}
	\begin{split}
	\prob(|R_j(x)| \geq \epsilon | \bbeta) \leq 2 \exp(- \frac{ \epsilon^2}{ 2  \sigma^2 \bpsi_{\supp^c}(x)_j^2 C^2 })
	\end{split}
	\end{align}
	Let $\max_{x, j} \bpsi(x)_j = \overline\bpsi$, then
	\begin{align}
	\begin{split}
	\prob(|R_j(x)| \geq \epsilon | \bbeta) \leq 2 \exp(- \frac{ \epsilon^2}{ 2  \sigma^2 \overline\bpsi^2 C^2 })
	\end{split}
	\end{align}
	Note that given $\bbeta$, $R_j(x) | \bbeta, \forall x \in D$ are mutually independent. Also note that $\prob(\frac{1}{N}\sum_{x \in D} |R_j(x)| \geq \epsilon) \geq \prob(\frac{1}{N}|\sum_{x \in D} R_j(x)| \geq \epsilon) $, we get:
	\begin{align}
	\begin{split}
	\prob(\frac{1}{N}|\sum_{x \in D} R_j(x)| \geq \epsilon | \bbeta) \leq 2 \exp(- \frac{N \epsilon^2}{ 2  \sigma^2 \overline\bpsi^2 C^2 })
	\end{split}
	\end{align}
	Again taking a union bound across $j \in \supp^c$, we get 
	\begin{align}
	\begin{split}
	\prob(\frac{1}{N} \| \sum_{x \in D} \bpsi_{\supp^c}(x) \bgamma_{\supp}^\T(x) \bbeta_{\supp}  \|_{\infty} \geq \epsilon | \bbeta) \leq 2 (n-d) \exp(- \frac{ N \epsilon^2}{ 2  \sigma^2 \overline\bpsi^2 C^2 })
	\end{split}
	\end{align}
	Computing the expectation with respect to $\bbeta$, we get
	\begin{align}
	\begin{split}
	\prob(\frac{1}{N} \| \sum_{x \in D} \bpsi_{\supp^c}(x) \bgamma_{\supp}^\T(x) \bbeta_{\supp}  \|_{\infty} \geq \epsilon) \leq 2 (n-d) \exp(- \frac{ N \epsilon^2}{ 2  \sigma^2 \overline\bpsi^2 C^2 })
	\end{split}
	\end{align}
	Thus $ \frac{1}{N}\| \sum_{x \in D} \bpsi_{\supp^c}(x) \bgamma_{\supp}^\T(x) \bbeta_{\supp} \|_{\infty}$ concentrates around $\epsilon$ as long as $N = \mathcal{O}(\frac{\sigma^2 \overline\bpsi^2 C^2  \log (n-d)}{\epsilon^2})$.
	
	Following the arguments similar to the previous proof, we can write that
	\begin{align}
		\begin{split}
		\prob(\frac{1}{N} \| \sum_{y \in D} \bpsi_{\supp}(y) \bgamma_{\supp}^\T(y) \bbeta_{\supp}  \|_{\infty} \geq \epsilon) \leq 2 d \exp(- \frac{ N \epsilon^2}{ 2  \sigma^2 \overline\bpsi^2 C^2 })
		\end{split}
	\end{align}
	Thus $ \frac{1}{N}\| \sum_{y \in D} \bpsi_{\supp}(x) \bgamma_{\supp}^\T(x) \bbeta_{\supp} \|_{\infty}$ concentrates around $\epsilon$ as long as $N = \mathcal{O}(\frac{\sigma^2 \overline\bpsi^2 C^2  \log d}{\epsilon^2})$.
\end{proof}

\section{Proof of Lemma \ref{lem:psi beta tildepsi}}
\label{proof:psi beta tildepsi}
\paragraph{Lemma \ref{lem:psi beta tildepsi}}
\emph{For sufficiently large $r$ and $\epsilon > 0$, we have
	\begin{align}
	\begin{split}
	&\prob(\frac{1}{N}\| (\sum_{y \in D} \bpsi_{\supp}(y) \sum_{j \ne i} \sum_{k=r+1}^\infty \beta_{ijk}^* \widetilde{\psi}(y_i, y_j)) \|_{\infty} \geq 2 \epsilon) \leq 2 d \exp(- \frac{N \epsilon^2}{2 \overline{\bpsi}^2 d \delta }  )
	\end{split}
	\end{align}
	as well as
	\begin{align}
	\begin{split}
	&\prob(\frac{1}{N}\| (\sum_{y \in D} \bpsi_{\supp^c}(y) \sum_{j \ne i} \sum_{k=r+1}^\infty \beta_{ijk}^* \widetilde{\psi}(y_i, y_j)) \|_{\infty} \geq 2 \epsilon) \leq 2 (n-d) \exp(- \frac{N \epsilon^2}{2 \overline{\bpsi}^2 d \delta }  )
	\end{split}
	\end{align}
	where $\delta > 0$ is an arbitrary constant which depends on $r$.	
}
\begin{proof}
	\begin{align}
	\begin{split}
	\| (\sum_{y \in D} \bpsi_{\supp}(y) \sum_{j \ne i} \sum_{k=r+1}^\infty \beta_{ijk}^* \widetilde{\psi}(y_i, y_j)) \|_{\infty} \leq \| (\sum_{y \in D} \bpsi_{\supp}(y) \sum_{j \ne i} \sum_{k=r+1}^\infty \beta_{ijk}^* \psi_k(y_i, y_j)) \|_{\infty} + \\
	\| (\sum_{y \in D} \bpsi_{\supp}(y) \sum_{j \ne i} \sum_{k=r+1}^\infty \beta_{ijk}^* \gamma_k(y_i, y_j))  \|_{\infty}
	\end{split}
	\end{align}
	We will bound both the terms separately. 
	\paragraph{Bound on $ \frac{1}{N}\| (\sum_{y \in D} \bpsi_{\supp}(y) \sum_{j \ne i} \sum_{k=r+1}^\infty \beta_{ijk}^* \psi_k(y_i, y_j)) \|_{\infty} $.}
		Let $R \triangleq (\sum_{y \in D} \bpsi_{\supp}(y) \sum_{j \ne i} \sum_{k=r+1}^\infty \beta_{ijk}^* \psi_k(y_i, y_j)) $ be a random variable. Let $l$th entry of $R$ be $R_l$, i.e., $R_l =  (\sum_{y \in D} \bpsi_{\supp}(y)_l \sum_{j \ne i} \sum_{k=r+1}^\infty \beta_{ijk}^* \psi_k(y_i, y_j)) $.
		
		\begin{align}
		\begin{split}
		|(\sum_{y \in D} \bpsi_{\supp}(y)_l \sum_{j \ne i} \sum_{k=r+1}^\infty \beta_{ijk}^* \psi_k(y_i, y_j))| &\leq  (\sum_{y \in D} |\bpsi_{\supp}(y)_l| \sum_{j \ne i} \sum_{k=r+1}^\infty |\beta_{ijk}^*| |\psi_k(y_i, y_j)|) \\
		&\leq (\sum_{y \in D} |\overline\bpsi| \sum_{j \ne i} \sum_{k=r+1}^\infty |\beta_{ijk}^*| |\overline\psi |) \\
		&\leq N  |\overline\psi |^2 \sum_{j \ne i} \sum_{k=r+1}^\infty |\beta_{ijk}^*|
		\end{split}
		\end{align}
		
		Let $[\alpha_{ijk}]_{k=1}^\infty$ be a convergent series with positive entries such that $|\bbeta_{ijk}^*| \leq \alpha_{ijk} , \forall k \geq r+1$. To give an example, let $\beta_{ijk}^*$ be the Fourier coefficients, then we can choose $\alpha_{ijk} = \frac{D}{k^2}$ for some $D > 0$. Then it holds that $|\beta_{ijk}| \leq \frac{D}{k^2}, \forall k \geq r+1$ for sufficiently large $r$.
		\begin{lemma}
			\label{lem:convergent_series}
			For sufficiently large $r$, $\sum_{k=r+1}^\infty |\beta_{ijk}^*| \leq \delta $ for any $\delta > 0$.
		\end{lemma} 
		\begin{proof}
			The tail sum of a convergent series goes to $0$, i.e., for sufficiently large $r$,
			\begin{align}
			\begin{split}
			\sum_{k=r+1}^\infty \alpha_{ijk}^* \leq \delta
			\end{split}
			\end{align} 
			It readily follows that $\sum_{k=r+1}^\infty |\beta_{ijk}^*| \leq \delta $.
		\end{proof}
		Using Lemma \ref{lem:convergent_series} and taking $\delta = \frac{\epsilon}{d |\overline{\psi}|^2 }$, it follows that:
		\begin{align}
		\begin{split}
		|(\sum_{y \in D} \bpsi_{\supp}(y)_l \sum_{j \ne i} \sum_{k=r+1}^\infty \beta_{ijk}^* \psi_k(y_i, y_j))| &\leq N \epsilon
		\end{split}
		\end{align}
		This ensures that for large enough $r$, we have $ \\\frac{1}{N} \| (\sum_{y \in D} \bpsi_{\supp}(y) \sum_{j \ne i} \sum_{k=r+1}^\infty \beta_{ijk}^* \psi_k(y_i, y_j)) \|_{\infty} \leq  \epsilon $.

	\paragraph{Bound on $ \frac{1}{N}	\| (\sum_{y \in D} \bpsi_{\supp}(y) \sum_{j \ne i} \sum_{k=r+1}^\infty \beta_{ijk}^* \gamma_k(y_i, y_j))  \|_{\infty} $.}
		Let $R$ be a random variable such that $R \triangleq \bpsi_{\supp}(y) \sum_{j \ne i} \sum_{k=r+1}^\infty \beta_{ijk}^* \gamma_k(y_i, y_j)$. Let the $l$-th entry of $R$ be $R_l = \bpsi_{\supp}(y)_l \sum_{j \ne i} \sum_{k=r+1}^\infty \beta_{ijk}^* \gamma_k(y_i, y_j)$. For given $\bbeta^*$, $R_l$ is a zero mean sub-Gaussian random variable with variance proxy $\bpsi_{\supp}(y)^2 \sum_{j\ne i} \sum_{k=r+1}^{\infty} (\beta_{ijk}^*)^2$. Note that since $\alpha_{ijk}$ is convergent series with positive entries, $\alpha_{ijk}^2$ is also convergent. This means that using the tail sum of a convergent series 
		\begin{align}
		\begin{split}
		\sum_{k=r+1}^{\infty} \alpha_{ijk}^2 \leq \delta  
		\end{split}
		\end{align}
		for sufficiently large $r$ and any $\delta > 0$. It follows that
		\begin{align}
		\begin{split}
		\sum_{k=r+1}^{\infty} (\beta_{ijk}^*)^2 \leq \delta  
		\end{split}
		\end{align}
		Also, note that $R_l | \bbeta^*$ are mutually independent for $y \in D$. Thus
		\begin{align}
		\begin{split}
		\prob(\frac{1}{N} |R_l| \geq \epsilon  | \bbeta^*) \leq 2 \exp(- \frac{N \epsilon^2}{2 \overline{\bpsi}^2 d \delta }  )
		\end{split}
		\end{align} 
		Taking a union bound across $l \in \supp$, we get
		\begin{align}
		\begin{split}
		\prob(\frac{1}{N}	\| (\sum_{y \in D} \bpsi_{\supp}(y) \sum_{j \ne i} \sum_{k=r+1}^\infty \beta_{ijk}^* \gamma_k(y_i, y_j))  \|_{\infty} \geq \epsilon  | \bbeta^*) \leq 2 d \exp(- \frac{N \epsilon^2}{2 \overline{\bpsi}^2 d \delta }  )
		\end{split}
		\end{align} 
		Computing the expectation with respect to $\bbeta^*$, we get
		\begin{align}
		\begin{split}
		\prob(\frac{1}{N}	\| (\sum_{y \in D} \bpsi_{\supp}(y) \sum_{j \ne i} \sum_{k=r+1}^\infty \beta_{ijk}^* \gamma_k(y_i, y_j))  \|_{\infty} \geq \epsilon  ) \leq 2 d \exp(- \frac{N \epsilon^2}{2 \overline{\bpsi}^2 d \delta }  )
		\end{split}
		\end{align} 
		where $\delta$ is an arbitrary constant which depends on $r$.

	Combining above results,
	\begin{align}
	\begin{split}
	\prob(\frac{1}{N}\| (\sum_{y \in D} \bpsi_{\supp}(y) \sum_{j \ne i} \sum_{k=r+1}^\infty \beta_{ijk}^* \widetilde{\psi}(y_i, y_j)) \|_{\infty} \geq 2 \epsilon) \leq 2 d \exp(- \frac{N \epsilon^2}{2 \overline{\bpsi}^2 d \delta }  )
	\end{split}
	\end{align}
	Thus if $N = \mathcal{O}(\frac{ \overline{\bpsi}^2 d \delta \log d}{\epsilon^2})$ is sufficient for $\frac{1}{N}\| (\sum_{y \in D} \bpsi_{\supp}(y) \sum_{j \ne i} \sum_{k=r+1}^\infty \beta_{ijk}^* \widetilde{\psi}(y_i, y_j)) \|_{\infty}$ to concentrate around $2\epsilon$.	
	Using the similar argument as the previous proof, we can write
	\begin{align}
	\begin{split}
	\prob(\frac{1}{N}\| (\sum_{y \in D} \bpsi_{\supp^c}(y) \sum_{j \ne i} \sum_{k=r+1}^\infty \beta_{ijk}^* \widetilde{\psi}(y_i, y_j)) \|_{\infty} \geq 2 \epsilon) \leq 2 (n-d) \exp(- \frac{N \epsilon^2}{2 \overline{\bpsi}^2 d \delta }  )
	\end{split}
	\end{align}
	where $\delta$ could be arbitrarily small depending on $r$. Thus if $N = \mathcal{O}(\frac{ \overline{\bpsi}^2 d \delta \log (n-d)}{\epsilon^2})$ is sufficient for $\frac{1}{N}\| (\sum_{y \in D} \bpsi_{\supp^c}(y) \sum_{j \ne i} \sum_{k=r+1}^\infty \beta_{ijk}^* \widetilde{\psi}(y_i, y_j)) \|_{\infty}$ to concentrate around $2\epsilon$.
\end{proof}

\section{Proof of Lemma \ref{lem:l2 norm bounds}}
\label{proof:l2 norm bounds}
\paragraph{Lemma \ref{lem:l2 norm bounds}}
	\emph{For some $\epsilon > 0$,
	\begin{align}
	\begin{split}
	&\prob(\|\frac{1}{N} (\sum_{x \in D} \bpsi_{\supp}(x) \bgamma_{\supp}^\T(x) \bbeta_{\supp}) \|_2 \geq \epsilon) \leq 2 d \exp( \frac{- N \epsilon^2}{2 \sigma^2 d\widebar\psi^2C^2})
	\end{split}
	\end{align}
	and
	\begin{align}
	\begin{split}
	&\prob(\|\frac{1}{N} (\sum_{x \in D} \bpsi_{\supp}(x) \sum_{j \ne i} \sum_{k=r+1}^\infty \beta_{ijk}^* \widetilde{\psi}_k(x_i, x_j)) \|_2 \geq \epsilon) \leq 2 d \exp( \frac{- N \epsilon^2}{2 d\widebar\psi^2 d \delta})
	\end{split}
	\end{align}
	where $\delta > 0$ could be an arbitrary small constant which depends on $r$.}
\begin{proof}
	Following the notations used in Lemma \ref{lem:psi gamma beta}, we know that
	\begin{align}
	\begin{split}
		\prob(| \frac{1}{N} \sum_{x \in D} R_j(x) | \geq \epsilon \mid \bbeta) \leq 2 \exp(\frac{- N \epsilon^2}{2 \sigma^2 \widebar\psi^2 C^2})
	\end{split}
	\end{align}
	We take $\epsilon = \frac{\epsilon}{\sqrt{d}}$ and apply a union bound over $j \in \supp$, we get
	\begin{align}
	\begin{split}
		\prob(\|\frac{1}{N} (\sum_{x \in D} \bpsi_{\supp}(x) \bgamma_{\supp}^\T(x) \bbeta_{\supp}) \|_2 \geq \epsilon \mid \bbeta) \leq 2 d \exp(\frac{- N \epsilon^2}{2 \sigma^2 d\widebar\psi^2 C^2})
	\end{split}
	\end{align} 
	Computing the expectation with respect to $\bbeta$,
	\begin{align}
		\begin{split}
		\prob(\|\frac{1}{N} (\sum_{x \in D} \bpsi_{\supp}(x) \bgamma_{\supp}^\T(x) \bbeta_{\supp}) \|_2 \geq \epsilon) \leq 2 d \exp(\frac{- N \epsilon^2}{2 \sigma^2 d\widebar\psi^2 C^2})
		\end{split}
	\end{align}
	Similarly, using notations used in Lemma \ref{lem:psi beta tildepsi}, we know that
	\begin{align}
	\begin{split}
		\prob(\frac{1}{N}|R_l| \geq \epsilon \mid \bbeta^* ) \leq 2 \exp(\frac{ - N \epsilon^2 }{ 2 \widebar\psi^2 d \delta})
	\end{split}
	\end{align}
	where $\delta > 0$ is an arbitrary parameter which depends on $r$. We take $\epsilon = \frac{\epsilon}{\sqrt{d}}$ and union bound over $l \in \supp$ and get
	\begin{align}
	\begin{split}
		\prob(\|\frac{1}{N} (\sum_{x \in D} \bpsi_{\supp}(x) \sum_{j \ne i} \sum_{k=r+1}^\infty \beta_{ijk}^* \widetilde{\psi}_k(x_i, x_j)) \|_2 \geq \epsilon | \bbeta^*) \leq 2 d \exp( \frac{- N \epsilon^2}{2 d\widebar\psi^2 d \delta})
	\end{split}
	\end{align}
	Taking expectation with respect to $\bbeta^*$, we get
	\begin{align}
	\begin{split}
		\prob(\|\frac{1}{N} (\sum_{x \in D} \bpsi_{\supp}(x) \sum_{j \ne i} \sum_{k=r+1}^\infty \beta_{ijk}^* \widetilde{\psi}_k(x_i, x_j)) \|_2 \geq \epsilon ) \leq 2 d \exp( \frac{- N \epsilon^2}{2 d\widebar\psi^2 d \delta})
	\end{split}
	\end{align}
\end{proof}

\section{Application - Identifying Countries Influencing Potato Trade in Europe}
\label{sec:real world data}

\begin{figure*}[!ht]
	\begin{center}
		\includegraphics[scale=0.5]{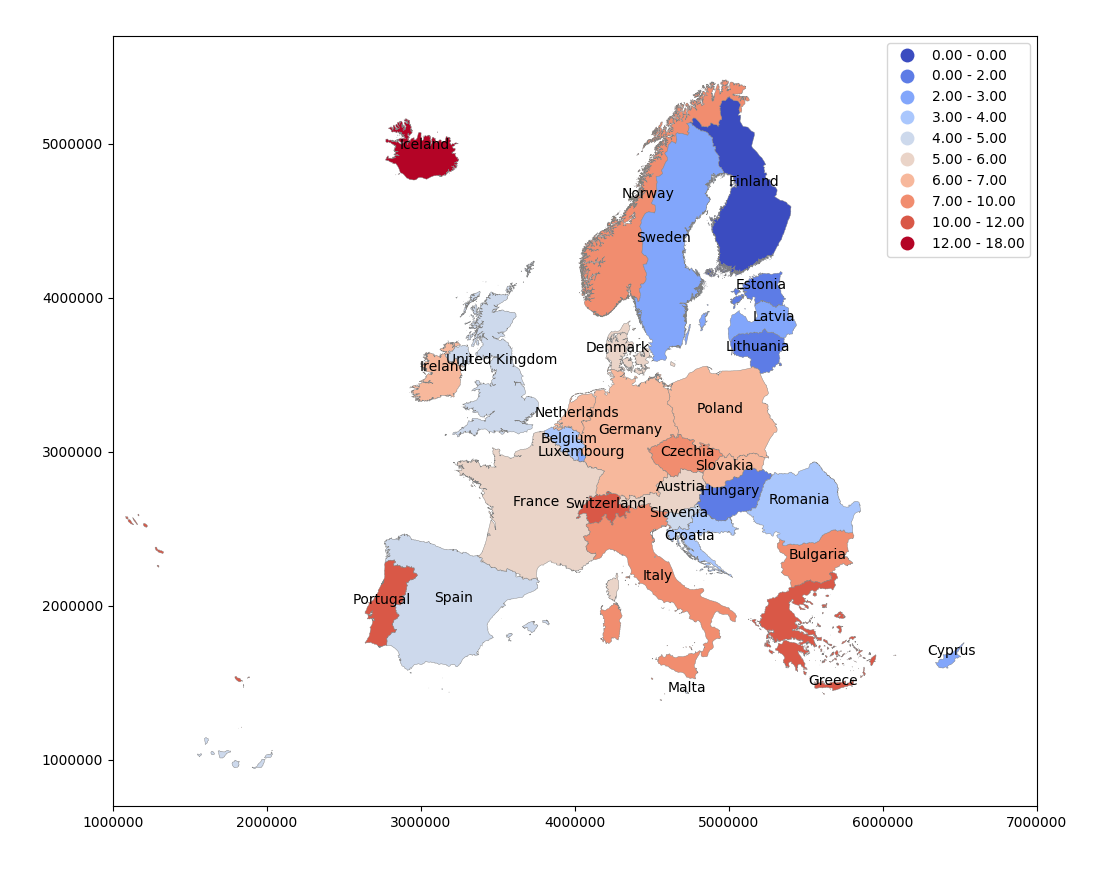}
		\caption{Influence of various European countries on potato trade between $1973-2019$. The influence of a country is measured by the number of its out-neighbors in the learnt game.}
		\label{fig:trading}
	\end{center}
\end{figure*}

In this section, we performed a computational experiment on real world data to demonstrate the effectiveness of our method. We used our method to identify a set of European countries which have influenced potato trade within Europe between $1973-2019$. The trading interactions among various countries can be modeled as a non-parametric strategic game among self-interested countries.

Our experiments were conducted on the publicly available potato trade data from \url{https://ec.europa.eu/eurostat/web/agriculture/data/database}. We extracted the potato trade data among $n = 31$ European countries between $1973$ to $2019$. Each training sample corresponds to a year, and thus, we had $T = 47$ samples for our experiments. The dataset contains the country-wise production volume of potato for $47$ years which we treated as the action $x$ of players. The dataset also contains country-wise trade prices of potato for $47$ years which we treated as payoff $\widetilde{u}_i(x)$ for each country. We chose $r = 16$ basis functions for the Fourier series expansion of each pairwise utility function to run our experiments. In particular, for each pairwise utility function $\widebar{u}_{ij}(x_i, x_j), \forall i,j \in \seq{n}, i\ne j$,  we chose the basis functions from the following set: 
\begin{align}
\begin{split}
	\psi_k(x_i, x_j) \in \bigcup_{l, m \in \{1, 2\}} &\{ \cos(2\pi l x_i) \cos(2 \pi m x_j), \cos(2\pi l x_i) \sin(2 \pi m x_j), \sin(2\pi l x_i) \cos(2 \pi m x_j),\\
	& \sin(2\pi l x_i) \sin(2 \pi m x_j) \}
\end{split}
\end{align}
We computed the influence of the countries by first learning the global graphical game and then computing the number of out-neighbors for each country in the learnt game. The results are shown in Figure \ref{fig:trading}. We observe that many of the central and southern European countries along with Iceland and Portugal have had high influence in the European potato market .

\end{document}